\documentclass[a4paper]{article}
\usepackage{amsthm}
\usepackage{amsmath}
\usepackage{amssymb}
\usepackage{latexsym}
\usepackage{eucal, color}

\newtheorem{theorem}{Theorem}[section]

\newtheorem{lemma}[theorem]{Lemma}

\newtheorem{remark}[theorem]{Remark}

\newtheorem{assumption}[theorem]{Assumption}
\begin{document}

\title{Behavioural investors in conic market models\thanks{Supported
by the ``Lend\"ulet'' grant LP 2015-6/2015 of the
Hungarian Academy of Sciences and by the NKFIH (National Research, Development and Innovation Office, Hungary) 
grant KH 126505. We thank an anonymous referee for useful comments. The paper is
dedicated to Yuri M. Kabanov, hoping that he will be satisfied with
the generality of the considered model class.}}
\author{Huy N. Chau \and Mikl\'os R\'asonyi}
\date{\today}
\maketitle

\begin{abstract}
We treat a fairly broad class of financial models which includes markets with proportional
transaction costs. We consider an investor with cumulative prospect theory preferences and a non-negativity
constraint on portfolio wealth. The existence of an optimal strategy is shown in this context in 
a class of generalized strategies.
\end{abstract}

\section{Introduction}

In this paper we continue the investigations of \cite{sicon} where behavioural
investors were studied in a model with price impact. In the current work
we treat the case of conic models, see \cite{ks}, which subsume foreign exchange markets as well as multi-asset
markets with proportional transaction costs.

The mathematical difficulty stems from the fact that behavioural preferences lack
concavity and involve probability distortions, see \cite{kt}, \cite{quiggin}, \cite{tk}. Hence, instead of almost sure techniques,
we need to employ weak convergence in the arguments. In Theorem \ref{main} below we
establish the existence of optimizers in a suitable class of generalized strategies.
We rely on results of \cite{jakubowski}, see Theorem \ref{altalanos} below.

In Section \ref{harom} we present our model. In Section \ref{ins}
we construct optimal strategies for investment 
problems with behavioural preferences. Section \ref{negy} collects auxiliary material.

\section{Conic market model}\label{harom}

We will assume throughout the paper that trading takes place continuously
in the time interval $[0,1]$. Let $(\Omega, \mathcal{F}, (\mathcal{H}_t)_{t\in [0,1]}, P)$ be a filtered 
probability space, where the filtration is complete and right-continuous, $\mathcal{H}_0$ is trivial.
The notation $EX$ will refer to the expectation of the 
random variable $X$. If there is ambiguity about the probability measure then $E_QX$ will denote
the expectation of $X$ under the probability $Q$. Similarly, $\mathrm{Law}(X)$ denotes the law of
$X$ and $\mathrm{Law}_Q(X)$ refers to its law under $Q$.
When $x,y$ are vectors in the same Euclidean space then the concatenation $xy$ denotes their
scalar product, $|x|$ is the Euclidean norm.

In the sequel we will need that the filtration is of a specific type and that the probability space 
is large enough. 
\begin{assumption}\label{u}
	There exists a c\`adl\`ag $\mathbb{R}^m$-valued process $Y$
	with independent increments such that $\mathcal{H}_t$ is the $P$-completion of $\sigma(Y_u,\, 0\leq u\leq t)$,
	for $t\in [0,1]$.
\end{assumption}

For $m\in\mathbb{N}$, we denote by $\mathcal{D}^m$ the space of $\mathbb{R}^m$-valued RCLL functions on $[0,1]$ 
equipped with Skorohod's topology, see Chapter 3 of \cite{billingsley}.

\begin{remark}\label{ipszilon}{\rm The Borel-field of 
		$\mathcal{D}^m$ is generated by the coordinate 
		mappings $x\in\mathcal{D}^m\to x(t)\in\mathbb{R}^m$, $t\in [0,1]$, see Theorem 12.5 of 
		\cite{billingsley}. It follows that the function $\omega\in \Omega\to Y(\omega)\in\mathcal{D}^m$
		is a random variable and so is $\omega\in \Omega\to ^t Y(\omega)\in\mathcal{D}^m$, for
		all $t\in [0,1]$, where $^t Y$ is the process defined as
		$(^t Y)_u=Y_u 1_{[0,t)}+Y_t 1_{[t,1]}$, $u\in [0,1]$. 
		Furthermore, $\mathcal{H}_t=\sigma(^t Y)$, for all $t\in [0,1]$.}
\end{remark}
\begin{assumption}\label{ghj}
	There exists a random variable $U$ that is uniformly distributed on $[0,1]$ and independent
	of $\mathcal{H}_1$. 
\end{assumption}


Let us define the augmented filtration $\mathcal{F}_t:=\mathcal{H}_t\vee \sigma(U)$, $t\in [0,1]$. 
Standard arguments show that $\mathcal{F}_t$, $t\in [0,1]$ also satisfies
the usual hypotheses of completeness and right-continuity. 

We now recall the market model presented in Subsection 3.6.3 of \cite{ks}.
Let $\xi^k_t$, $t\in [0,1]$, be $\mathcal{H}$-adapted $\mathbb{R}^d$-valued processes for each $k\in\mathbb{N}$
such that, for a.e. $\omega$ and for all $t$, only finitely many terms of the sequence $\xi^k_t(\omega)$, $k\in\mathbb{N}$ differ from $0$. 
Let $G_t(\omega)$, $t\in [0,1]$, $\omega\in\Omega$ denote the polyhedral cone generated by
$\xi_t^k(\omega)$, $k\in\mathbb{N}$. We assume that $\mathbb{R}^d_+\subset G_t$ a.s. for each $t\in [0,1]$. 
Let the dual cones be defined by $G_t^*(\omega):=\{x\in\mathbb{R}^d:\ xy\geq 0\mbox{ for all }y\in G_t(\omega)\}$.
We imagine that $G_t(\omega)$ represents the set of solvent positions in $d$ financial assets
at time $t$ in the state of the world $\omega\in\Omega$.

\begin{assumption}\label{dual}
There is a family of $\mathcal{H}$-adapted \emph{continuous} processes $\zeta^k_t$, $t\in [0,1]$, $k\in\mathbb{N}$
such that $G_t^*(\omega)$ is generated by $\zeta_t^k(\omega)$, $k\in\mathbb{N}$ and only finitely many terms of this
sequence differ from $0$, for a.e. $\omega$ and for every $t$.
\end{assumption} 

Although the dual generators $\zeta^k$, $k\in\mathbb{N}$ are assumed to be continuous processes, 
the above assumption allows them
to depend on a driving process $Y$ with possibly discontinuous paths (consider e.g.\ a stochastic volatility model with jumps in the
volatility). 

The following assumption requires that there is efficient friction in the market, see page 158 of \cite{ks}.

\begin{assumption}\label{ef}
Fore each $t\in [0,1]$ and for a.e.\ $\omega\in\Omega$, $\mathrm{int}\, G_t^*(\omega)\neq\emptyset$.
\end{assumption}

Let $\mathfrak{D}$ denote the set of $\mathcal{H}$-adapted martingales $Z_t$, $t\in [0,1]$ such that $Z_t\in \mathrm{int}\, G_t^*$ and
$Z_{t-}\in\mathrm{int}\, G_t^*$ a.s.\ for each $t\in [0,1]$. The next assumption 
is essentially condition $\mathbf{B}$ on page 160 of \cite{ks}, it stipulates that there is a
rich enough class of objects in $\mathfrak{D}$. 

\begin{assumption}\label{mcps}
Assume that $\mathfrak{D}$ is nonempty. For each $s\in [0,1]$, and for each $\mathcal{H}_s$-measurable random variable $\xi$ if 
$\xi Z_s\geq 0$ for all $Z\in \mathfrak{D}$ 
then $\xi\in G_s$ a.s. 
\end{assumption}

For an $\mathbb{R}^d$-valued $\mathcal{F}_t$-adapted c\`adl\`ag process $X$ with bounded variation we denote by $||X||$ its
total variation process (scalar-valued) and let $\dot{X}$ denote the pathwise Radon-Nykodim
derivative of $X$ with respect to $||X||$, this can be chosen to be an $\mathbb{R}^d$-valued process. 
Let $\mathcal{X}^0$ denote the family of $\mathcal{F}$-adapted processes with bounded variation $X$
such that $X_0 = 0$ and $\dot{X}_t\in -G_t$ a.s.\ for all $t\in [0,1]$. 
These processes represent the evolution of portfolio
positions in a self-financing way, starting from initial position $0$. 

For each integer $k\geq 1$, 
consider $\mathcal{C}^k$, the space
of $\mathbb{R}^k$-valued continuous functions on the unit interval. This is a separable Banach space with the supremum norm.
Let $\mathfrak{M}^{2d}$ denote the Banach space of $2d$-tuples of finite signed measures on $\mathcal{B}([0,1])$. This is the dual space 
of $\mathcal{C}^{2d}$ with the total variation norm, henceforth denoted by $||\cdot||_1$.
However, in the seuqel we
equip $\mathfrak{M}^{2d}$ with the weak-$*$ topology in the natural dual pairing between $\mathcal{C}^{2d}$ and $\mathfrak{M}^{2d}$.

\begin{remark}
{\rm Let us notice that if 
$X\in\mathcal{X}^0$ then, for each $\omega\in\Omega$, 
$X(\omega)$ can be naturally identified with an element of $\mathfrak{M}^{2d}$. Indeed, we may consider 
$$
\overline{X}^{2j-1}(\omega)(A):=\int_A(\dot{X}^j_t)^+ d||X||_t(\omega),\ A\in\mathcal{B}([0,1]),\ j=1,\ldots,d,
$$
and
$$
\overline{X}^{2j}(\omega)(A):=\int_A(\dot{X}^j_t)^- d||X||_t(\omega),\ A\in\mathcal{B}([0,1]),\ j=1,\ldots,d.
$$
Furthermore, we claim that the mapping $\overline{X}:\Omega\to \mathfrak{M}^{2d}$ is $\mathcal{F}_1$-measurable.
Indeed, it suffices to show that for each continuous $\phi:[0,1]\to\mathbb{R}^d$,
the mapping $\omega\to\int_0^1 \phi(u)(\dot{X}^j_u)^+ d||X||_u(\omega)$ is $\mathcal{F}_1$-measurable for each $j=1,\ldots,d$ (similarly for $(\dot{X}^j_u)^-$), which is clear since $X$ is c\`adl\`ag and adapted.
By similar arguments, $\omega\to ^t \overline{X}(\omega)$ is $\mathcal{F}_t$-measurable, for every $t\in [0,1]$,
where $^t\overline{X}(\omega)(A):=\overline{X}(\omega)(A\cap [0,t])$. We will identify $X$ with $\overline{X}$ in the
sequel: when we write $X$ it may refer to either the stochastic process or to the $\mathfrak{M}^{2d}$-valued
random variable. A similar identification of $^t X$
with $^t \overline{X}$ will also be used.}
\end{remark}


For each initial position $x\in G_0$, 
we furthermore define $\mathcal{A}(x):=\{X\in\mathcal{X}^0:\ x+X_t  \in G_t\mbox{ a.s.\ for all }t\in [0,1]\}$,
the portfolio value processes which never become insolvent.

\begin{remark}
{\rm Investment decisions will be based on the augmented filtration $\mathcal{F}$. It is pointed out in \cite{cr} that by using a uniform $U$ (independent of $\mathcal{H}_1$) for randomizing the strategies an investor can increase her satisfaction, however, further randomizations are pointless. See Remarks 22 and 23 of \cite{sicon} and Section 5 of \cite{cr} 
for detailed explanations. Unlike other studies, we assume that the ``dual process'' $Z$ is $\mathcal{H}$-adapted, since information from $U$ does not weaken market viability.}   
\end{remark}

We fix a function $\ell:\mathcal{D}^m\times\mathbb{R}^d\to\mathbb{R}$ (interpreted as a \emph{liquidation function})
which transfers the terminal portfolio position into cash. We assume that it is continuous.   
The liquidation value of a position $x\in\mathbb{R}^d$ is $\ell(Y,x)$ (so it depends on the
market situation via $Y$).


\section{Optimal investments}\label{ins}

For $z\in\mathbb{R}$ we denote $z^+:=\max\{z,0\}$, $z^-:=\max\{-z,0\}$.
Let $u_+,u_-: \mathbb{R}_+ \to \mathbb{R}_+$ be continuous, increasing functions such that $u_{\pm}(0)=0$.
Let $w_+,w_-: [0,1]\to [0,1]$ be continuous with $w_{\pm}(0)=0$, $w_{\pm}(1)=1$. Functions
$u_{\pm}$ express the agent's attitude towards gains and losses while $w_{\pm}$ are
functions distorting the probabilities of events, see \cite{tk}, \cite{cr}.

We define, for any random variable $X\geq 0$,
\begin{eqnarray}\nonumber
V_+(X):=\int_0^{\infty} w_+\left(P\left(
u_+\left(X\right)\geq y
\right)\right)dy,
\end{eqnarray}
and
\begin{eqnarray}\nonumber
V_-(X):=\int_0^{\infty} w_-\left(P\left(
u_-\left(X\right)\geq y
\right)\right)dy.
\end{eqnarray}
For each real-valued random variable $X$ with $V_+(X^+)<\infty$ we set 
\begin{equation}\nonumber
V(X):=V_+(X^+)-V_-(X^-). 
\end{equation}

\begin{assumption}\label{bou}
The function $u_+$ is bounded from above.
\end{assumption}

Assumption \ref{bou} could be substantially relaxed at the price of requiring stronger
assumptions about $\mathfrak{D}$ but this would significantly complicate the arguments. 
Let $W$ be an $\mathcal{H}_1$-measurable $d$-dimensional random variable representing a reference point for the investor
in consideration. 
Notice that under Assumption \ref{bou} the functional $V(\ell(Y,X_1-W))$ is well-defined for every $X\in\mathcal{A}(x)$.

The quantity $V(\ell(Y,X_1-W))$ expresses the satisfaction of an agent with CPT preferences when (s)he
has a portfolio process $X$, see \cite{jz,cr} for more detailed discussions. 
Positive $\ell(Y,X_1-W)$ means outperforming the benchmark $W$, negative $\ell(Y,X_1-W)$
means falling short of it. Doob's theorem implies that there is a measurable 
$h:\mathcal{D}^m\to \mathbb{R}^d$ such that
$W=h(Y)$.

We aim to find an optimal investment strategy, i.e.  $X^{\dagger}\in\mathcal{A}(x)$ with
\begin{equation*}\label{problem}
V(\ell(Y,X_1^{\dagger}-W))=\sup_{X\in\mathcal{A}(x)}V(\ell(Y,X_1-W)).
\end{equation*}

The next theorem is our main result on the existence of optimizers for behavioural investors in conic models.

\begin{theorem}\label{main} Let Assumptions \ref{u}, \ref{ghj}, \ref{dual}, \ref{ef}, \ref{mcps}  and \ref{bou} be valid.
Fix $x\in G_0$. There exists $X^{\dagger}\in
	\mathcal{A}(x)$ such that 
	$$V(\ell(Y,X_1^{\dagger}-W))=\sup_{X\in\mathcal{A}(x)}V(\ell(Y,X_1-W)).$$
\end{theorem}

\begin{remark}{\rm Let 
$u:\mathbb{R}^d\to\mathbb{R}$ be continuous and bounded from above.
The arguments in the proof below can also establish that there is 
$X^{\dagger}\in\mathcal{A}(x)$ such that 
$$
Eu(X^{\dagger}_1)=\sup_{X\in\mathcal{A}(x)}Eu(X_1).$$
}
\end{remark}

\begin{proof}[Proof of Theorem \ref{main}] 
Let $X(n)\in\mathcal{A}(x)$, $n\in\mathbb{N}$ be such that 
	$$
	V(\ell(Y,X_1(n)-W))\to\sup_{X\in\mathcal{A}(x)}V(\ell(Y,X_1-W)),\ n\to\infty.
	$$
Applying Lemma 3.6.4 of \cite{ks} to the set $\{X(n), n \in \mathbb{N}\}$ with the choice
 $\kappa := |x|$, there exists a probability measure $Q \sim P$ such that 
$\sup_{n \in \mathbb{N}} E_Q ||X(n)||_1 < \infty$\footnote{In \cite{ks}, $Z$ and $X$ are adapted to the same filtration 
$\mathcal{H}$. Here, we allow $X$ to be a $\mathcal{F}$-adapted process but this causes no problem.}.   
Let $c_n, n \in \mathbb{N}$ be an arbitrary sequence of positive real numbers converging to $0$. Letting $\varepsilon > 0$, 
the Markov inequality yields 
$$
\lim_{n \to \infty} Q(c_n||X(n)||_1\geq \varepsilon) \le \lim_{n \to \infty} c_n E_Q[||X(n)||_1]/{\varepsilon} = 0.
$$
In other words, $c_n || X(n) ||_1$ converges to $0$ in $Q$-probability and hence in $P$-probability as well by the equivalence of $Q$ and $P$. 
Lemma 3.9 of \cite{k} shows that the sequence of $\mathbb{R}$-valued random variables $\|X(n)\|_1$, $n\in\mathbb{N}$ is tight. 

For any $r >0$, the set $\{m\in\mathfrak{M}^{2d}:||m||_1\leq r\}$ is weak-$*$ compact by the Banach-Alaoglou theorem hence
the $\mathfrak{M}^{2d}$-valued sequence $X(n)$ is tight. So is the sequence $(X(n),Y)$. Applying Theorem \ref{altalanos}, there exist
a probability space	$(O,\mathcal{O},R)$ and
	$\mathfrak{M}^{2d}\times\mathcal{D}^m$-valued random variables 
	$(\tilde{X}(n),Y(n))$ that converge
	$R$-a.s. to $(X^*,Y^*)$ along a subsequence (for which we keep the same notation) and 
	$\mathrm{Law}_R(\tilde{X}(n),Y(n))=\mathrm{Law}(X(n),Y)$, $n\in\mathbb{N}$. 
By subtracting a further subsequence we may and will also assume that 
\begin{equation}\label{eq}
\tilde{X}_1(n)\to X^*_1\mbox{ in law as }n\to\infty.
\end{equation}

For each $k\in\mathbb{N}$, let $f_k:\mathcal{D}^m\to\mathcal{C}^d$ be such that $\zeta^k=f_k(Y)$.
Such functions exist by Doob's lemma. Passing to a further
	subsequence through a diagonal argument, we may and will assume that, for each $k\in\mathbb{N}$, 
$\zeta^k(n):=f_k(Y(n))\to \zeta^{*k}:=f_k(Y^*)$ $R$-a.s. in $\mathcal{C}^d$ when
$n\to\infty$ 
	by Lemma \ref{fonction} and by the fact that each $Y(n)$
	has the same law (on $\mathcal{D}^m$). Analogously, we may and will assume $W(n):=h(Y(n))\to W^*:=h(Y^*)$
	$R$-a.s. in $\mathbb{R}^d$.


Let us define the analogue of the functionals $V_{\pm}$, $V$, for non-negative random variables
$X$ on $(O,\mathcal{O},R)$.
\begin{eqnarray*}
	V_+^R(X):=\int_0^{\infty} w_+\left(R\left(
		u_+\left(X\right)\geq y
		\right)\right)dy,
\end{eqnarray*}
and
\begin{eqnarray*}
	V_-^R(X):=\int_0^{\infty} w_-\left(R\left(
	u_-\left(X\right)\geq y
	\right)\right)dy.
\end{eqnarray*}
For each real-valued random variable $X$ on $(O,\mathcal{O},R)$ with $V_+^R(X^+)<\infty$ we set 
\begin{equation*}
V^R(X):=V^R_+(X^+)-V^R_-(X^-). 
\end{equation*}
Assumption \ref{bou} and the reverse Fatou lemma imply that 
\begin{equation}\label{f}
V^R(\ell(Y^*,X^*_1-W^*))\geq \limsup_n V^R(\ell(Y(n),X_1(n)-W(n))),
\end{equation} so 
$V^R(\ell(Y^*,X^*_1-W^*))\geq \sup_{X\in\mathcal{A}(x)} V(\ell(Y,X_1-W))$. 


Let us invoke Lemma \ref{transfer} with the choice $\tilde{\phi}:=X^*$, $\tilde{H}:=Y^*$ and
$H:=Y$. We get a $\mathcal{F}_1$-measurable random element
	$X^{\dagger}:=\phi\in \mathfrak{M}^{2d}$ satisfying $\mathrm{Law}(X^{\dagger},Y)=\mathrm{Law}_R
	(X^*,Y^*)$. Let us fix $0\le t < u \le 1$. We recall that 
	$^t X(n)$ is independent from ${Y}_u - {Y}_t$, or equivalently, 
	$$
	\mathrm{Law}(^t X(n) ,{Y}_u - {Y}_t) = \mathrm{Law}(^t X(n)) \otimes 
	\mathrm{Law}({Y}_u - {Y}_t).$$
	By construction, $\mathrm{Law}(^t X(n),Y_u - Y_t)  = \mathrm{Law}_R(^t \tilde{X}(n),Y_u(n)-{Y}_t(n))$.
	This implies also 
	$$
	\mathrm{Law}_R(^t \tilde{X}(n),Y_u(n) - Y_t(n))  = \mathrm{Law}_R(^t \tilde{X}(n))\otimes \mathrm{Law}_R(Y_u(n) - {Y}_t(n)).
	$$
	Passing to the limit as $n\to\infty$,
	$$
	\mathrm{Law}_R(^t X^*,Y^*_u - Y^*_t)  = 
	\mathrm{Law}_R(^t X^*) \otimes \mathrm{Law}_R(Y^*_u - {Y}^*_t),$$
	which implies independence of $^t X^{\dagger}\in \mathfrak{M}^{2d}$ from $_t Y\in\mathcal{D}^m$ as well where
	$(_t Y)_s:= 0$ if $0\leq s\leq t$ and $(_t Y)_s:=Y_s-Y_t$, $t<s\leq 1$. 

Since $Y$ is clearly a 
	measurable function 
	of $(_t Y,^t Y)\in \mathcal{D}^m\times\mathcal{D}^m$, applying Lemma \ref{fuggi} with the choice $\mathfrak{b}:=_t Y$ and $\mathfrak{a}:=(U,^t Y)$ we get that
	$^t X^{\dagger}$ is $\mathcal{F}_t$-measurable, for all $t$.
	
The set $\mathcal{L}:=\{Z_1:\, Z\in\mathfrak{D}\}$ is a subset of the separable metric space $L^1(P)$ hence it is also separable. Let $\{Z^k_1$, $k\in\mathbb{N}\}$ be a countable dense subset of $\mathcal{L}$. 
For each $k \in \mathbb{N}$, there exist measurable functions  $g_{k,s}:\mathcal{D}^m\to \mathbb{R}^d$ such that $E[Z^k_1\vert\mathcal{H}_s]=g_{k,s}(Y)$. Let $\xi$ be an $\mathcal{H}_s$-measurable random variable. By the density of the family 
$\{Z^k_1, k \in\mathbb{N}\}$ and Assumption \ref{mcps}, if $\xi g_{k,s}(Y)\geq 0$ a.s. for each $k$ then $\xi\in G_s$ a.s. Indeed, let $Z$ be an arbitrary element of $\mathfrak{D}$ and $Z^{k_n}_1, n \in \mathbb{N}$ be a sequence in the dense subset such that $Z^{k_n}_1\to Z_1$ in $L^1(P)$, and hence, $E[Z^{k_n}_1\vert\mathcal{H}_s] \to E[Z_1\vert\mathcal{H}_s]$ in $L^1(P)$ as well. One can extract a subsequence $k_{n_l}, l \in \mathbb{N}$ along which almost sure convergence holds, i.e. $g_{k_{n_l},s}(Y) \to Z_s$, $P$-a.s. Therefore, the fact 
$\xi g_{k_{n_l},s}(Y)\geq 0$ a.s. for each $l$ implies $\xi Z_s\geq 0$ a.s. and then $\xi \in G_s$ a.s.\ by Assumption \ref{mcps}.

Fix $k \in \mathbb{N}$ for a moment. 
Since $X_s(n) \in G_s$, obviously $X_s(n) g_{k,s}(Y)\geq 0$ $P$-a.s. for each $n \in \mathbb{N}$. Hence, we obtain $\tilde{X}_s(n) g_{k,s}(Y(n)) \ge 0$, $R$-a.s. for all $n$. By construction, $\tilde{X}(n)$ tends to $X^*$ $R$-a.s. in $\mathfrak{M}^{2d}$ (equipped
with the weak-$*$ topology). 
Moreover, from the properties of weak convergence of probabilities on $\mathbb{R}$ we know that, for $R$-a.e. $\omega$, 
$\lim_{n \to \infty}\tilde{X}_s(n)(\omega) = X^*_s(\omega)$ for every 
$s \in [0,1] \setminus I(\omega)$ where $I(\omega)$ is a countable set. 
Fubini's theorem then implies that there is a fixed set $T$ of Lebesgue measure $0$ such that for $s\notin T$,
$\lim_{n \to \infty}\tilde{X}_s(n) = X^*_s$ $R$-a.e. By \eqref{eq} we may assume that $1\notin T$.

An application of Lemma \ref{fonction} gives $X^*_s g_{k,s}(Y^*) \ge 0$, $R$-a.s. for every $s \in [0,1] \setminus T$. Notice that $X^{\dagger}_s=j(U,Y)$ for some $j:[0,1]\times\mathcal{D}^m\to \mathbb{R}$ is $ \mathcal{B}([0,1])\otimes \mathcal{G}_s$-measurable
where $\mathcal{G}_s$ is generated by the coordinate mappings of $\mathcal{D}^m$ up to $s$.

This means that for
$$
B:=\cap_{k\in\mathbb{N}} \{(u,y):\ j(u,y)g_{k,s}(y)\geq 0 \}
$$
we have $[\mathrm{Leb}\times \mathrm{Law}(Y)](B)=1$. But then, for $\mathrm{Leb}$-a.e. $u$, for $\mathrm{Law}(Y)$-a.e. $y$, 
$$
j(u,y)g_{k,s}(y)\geq 0,\ k\in\mathbb{N},
$$
which implies $j(u,Y)Z^k_s\geq 0$ a.s. for $\mathrm{Leb}$-a.e. $u$ and for each $k\in\mathbb{N}$. Noting that $j(u,Y)$ is $\mathcal{H}_s$-measurable, Assumption \ref{mcps} gives $j(u,Y) \in G_s$, for $\mathrm{Leb}$-a.e. $u$. 
This means 
$X^{\dagger}_s\in G_s$ a.s.

Fix now some $t \in T$ and let $s_n, n \in \mathbb{N}$ be a sequence in $[0,1] \setminus T$ such that $s_n \downarrow t.$ Right-continuity implies that $X^{\dagger}_t \xi^k_t = \lim_{n \to \infty} X^{\dagger}_{s_n} \xi^k_{s_n} \ge 0$. We thus conclude that  $X^{\dagger}_s \in G_s$ a.s. for all $s \in [0,1]$.
	
	To prove $\dot{X}^{\dagger}_t \in -G_t$, it suffices to show 
that the integrals $\int_0^{\cdot}\zeta^k_t dX_t^{\dagger}$, $k\in\mathbb{N}$ are 
non-increasing, by Lemma \ref{dot}. Indeed, from $\dot{X}_t(n) \in -G_t$ for all $t \in [0,1]$, it follows that
	$$ \int_s^t{\zeta^k_u dX_u(n)} \le 0, P\mbox{-a.s.} $$
	for any $0 \le s < t \le 1$. Lemma \ref{n} gives us 
	$$\int_s^t{\zeta^k_u(n) d\tilde{X}_u(n)} \le 0, R\mbox{-a.s.}   $$
	Again, the facts that $\tilde{X}(n)$ tends to $X^*$ $R$-a.s. in $\mathfrak{M}^{2d}$ and $\zeta^k(n):=f_k(Y(n))$ tends to $\zeta^{*k}:=f_k(Y^*)$ 
$R$-a.s. in $\mathcal{C}^{2d}$ imply
	$$\int_s^t{\zeta^{*k}_u dX^*_u} \le 0, R\mbox{-a.s.} $$
	Thus, 
	$$\int_s^t{\zeta^{k}_u dX^{\dagger}_u} \le 0, P\mbox{-a.s.} $$
	that is, $\int_0^{\cdot}\zeta^k_t dX_t^{\dagger}$ is non-increasing.

The previous arguments show $X^{\dagger}\in\mathcal{A}(x)$.	As $\mathrm{Law}(X^{\dagger},Y)=\mathrm{Law}_R
	(X^*,Y^*)$,
	$$
	\mathrm{Law}_R\left(X^*_1-W^*\right)=
	\mathrm{Law}(X_1^{\dagger}-W),
	$$
	and \eqref{f} shows that $X^{\dagger}$
	is the maximizer we have been looking for.	
\end{proof}

\section{Auxiliary results}\label{negy} 

We denote by $\mathcal{B}(\mathbf{Z})$ the Borel-field of a topological space $\mathbf{Z}$.
A sequence of probabilities $\mu_k$, $k\in\mathbb{N}$ on $\mathcal{B}(\mathbf{Z})$ is said to be \emph{tight} if, 
for all $\varepsilon>0$,
there is a compact set $K(\varepsilon)\subset \mathbf{Z}$ such that, for all $k$, $\mu_k(\mathbf{Z}\setminus K(\varepsilon))<
\varepsilon$. Take 
$\mathbf{Z}:=\mathfrak{M}^{2d}\times \mathcal{D}^{m}$.

\begin{theorem}\label{altalanos}
	Let $\mu_k$, $k\in\mathbb{N}$
	be a tight sequence of measures on $\mathcal{B}(\mathbf{Z})$. Then there is a subsequence $k_j$, $j\in\mathbb{N}$ 
	and a probability space on which there exist $\mathbf{Z}$-valued random variables $\xi$, $\xi_j$, 
	with $\mathrm{Law}(\xi_j)=
	\mu_{k_j}$, $j\in\mathbb{N}$ and $\xi_j\to\xi$ a.s., $j\to\infty$. 
\end{theorem}
\begin{proof}
This follows as in Corollary 3 and Example 5 of \cite{sicon}, using results of \cite{jakubowski}, .
\end{proof}

\begin{remark}
	{\rm Note that the space $\mathbf{Z}$ is not metrizable so the well-known
		versions of Skorohod's representation theorem (see e.g. Lemma
		4.30 in \cite{k}) are not applicable.}
\end{remark}


\begin{lemma}\label{fuggi}
	Let $(A,\mathcal{A})$, $(B,\mathcal{B})$ be measurable spaces and $j:A\times B\to\mathbb{R}$
	a measurable mapping. Let $(\mathfrak{a},\mathfrak{b})$ be an $A\times B$-valued random variable.
	If $\sigma(j(\mathfrak{a},\mathfrak{b}),\mathfrak{a})$ is independent of $\mathfrak{b}$ then 
	$j(\mathfrak{a},\mathfrak{b})$ is
	$\sigma(\mathfrak{a})$-measurable.
\end{lemma}
\begin{proof} See Lemma 29 of \cite{sicon}.
\end{proof}

We also recall Th\'eor\`eme 1 of \cite{beksy}.

\begin{lemma}\label{fonction} Let $A,B$ be separable metric spaces and 
$\xi_n\in A$, $n\in\mathbb{N}$ a sequence of 
	random variables
	converging to $\xi\in A$ in probability such that $\mathrm{Law}(\xi_n)$ is the same for all $n$.
	Then for each measurable $h:A\to B$ the random variables $h(\xi_n)$ converge to $h(\xi)$
	in probability (hence also a.s. along a subsequence).
	\hfill $\Box$
\end{lemma}

\begin{lemma}\label{transfer} Let $B$ be a measurable space.
	Let $H,\tilde{H}$ be random elements in $B$ with identical laws, defined on the probability spaces $(\Xi,\mathcal{E},R)$,
	$(\tilde{\Xi},\tilde{\mathcal{E}},\tilde{R})$, respectively. Let $\tilde{\phi}$ be a random element in $\mathbf{Z}$, defined on $(\tilde{\Xi},\tilde{\mathcal{E}},\tilde{R})$. 
	Let $U$ be independent of $H$ with uniform law on $[0,1]$. There exists a measurable 
	function $f: B \times [0,1] \to \mathbf{Z}$ such that $\phi = f(H, U)$ satisfies $Law_R(H, \phi) = Law_{\tilde{R}}(\tilde{H}, \tilde{\phi})$.
\end{lemma}
\begin{proof} Notice that the topological space 
$\mathbf{Z}$ is the union of its closed, increasing subspaces 
$A_n$, $n\in\mathbb{N}$ which are Polish spaces (with appropriate metrics). Now use Lemma 31 of \cite{sicon}.
\end{proof}

We give a criterion of admissibility for $\dot{X}$.

\begin{lemma}\label{dot}
A $\mathcal{F}$-adapted process $X$ of bounded variation satisfying $\dot{X}_t \in -G_t$ for all $t \in [0,1]$ if and only if the integrals
$\int_0^{\cdot}\zeta^k_t dX_t$ are non-increasing, for all $k\in\mathbb{N}$.
\end{lemma}
\begin{proof}
Identical to the proof of Lemma 3.6.1 of \cite{ks}. 
\end{proof}

\begin{lemma}\label{n}
Let $Y,\tilde{Y}$ be c\`adl\`ag processes, $X, \tilde{X}$ bounded variation processes defined on two probability spaces $(\Xi,\mathcal{E},R)$,
$(\tilde{\Xi},\tilde{\mathcal{E}},\tilde{R})$, respectively. Assume that $(\tilde{Y},\tilde{X})$ has the same law as $(Y, X)$. 
Let $f:\mathcal{D}^m\to\mathcal{C}^d$ be measurable. 
Then for all $0\le s < t \le 1$, it holds that
\begin{equation}\label{it}
\mathrm{Law}_{\tilde{R}}\left(\int_s^t{f(\tilde{Y})_u\, d\tilde{X}_u}\right)  = \mathrm{Law}_R\left(\int_s^t{f(Y)_u\, dX_u}\right).
\end{equation} 
\end{lemma}
\begin{proof}
We approximate $f$ by step functions and then pass to the limit.
\end{proof}

\newcommand{\etalchar}[1]{$^{#1}$}

\end{document}